\documentclass[11pt,letterpaper,english]{article}

\usepackage{hyperref}
\hypersetup{breaklinks}

\usepackage{amsmath,amssymb,amsthm}
\usepackage{algorithmic}

\usepackage[ruled,vlined,linesnumbered]{algorithm2e}
\usepackage{pgf, tikz}
\usepackage{xspace, units}

\usepackage{lmodern}
\usepackage[T1]{fontenc}
\usepackage{textcomp}

\usepackage{fullpage}

\newcommand{\lref}[2][]{\hyperref[#2]{#1~\ref*{#2}}}

\newtheorem{definition}{Definition}
\newtheorem{problem}{Problem}
\newtheorem{theorem}{Theorem}
\newtheorem{proposition}[theorem]{Proposition}
\newtheorem{lemma}[theorem]{Lemma}
\newtheorem{corollary}[theorem]{Corollary}

\newtheorem{observation}[theorem]{Observation}

\newcommand{\classP}{{\sf P}}
\newcommand{\classNP}{{\sf NP}}
\newcommand{\classZPP}{{\sf ZPP}}

\newcommand{\NN}{\ensuremath{\mathbb{N}}}

\newcommand{\RR}{\ensuremath{\mathbb{R}}}
\renewcommand{\O}{{\mathcal O}}

\newcommand{\inductiveindependence}[1]{inductive independence number#1}

\renewcommand{\Pr}[1]{\mbox{\rm\bf Pr}\left[#1\right]}
\newcommand{\Ex}[1]{\mbox{\rm\bf E}\left[#1\right]}

\author{Martin Hoefer \and Thomas Kesselheim \and Berthold V\"ocking}

\title{Approximation Algorithms for Secondary Spectrum Auctions%
  \thanks{Department of Computer Science, RWTH Aachen University,
    Germany. \texttt{$\{$mhoefer,kesselheim,voecking$\}$@cs.rwth-aachen.de}. This
    work has been supported by DFG through UMIC Research Centre, RWTH
    Aachen University, and grant Ho 3831/3-1.}}

\date{}

\begin{document}

\maketitle
\begin{abstract}
  We study combinatorial auctions for the secondary spectrum
  market. In this market, short-term licenses shall be given to
  wireless nodes for communication in their local neighborhood. In
  contrast to the primary market, channels can be assigned to multiple
  bidders, provided that the corresponding devices are well separated
  such that the interference is sufficiently low. Interference
  conflicts are described in terms of a conflict graph in which the
  nodes represent the bidders and the edges represent conflicts such
  that the feasible allocations for a channel correspond to the
  independent sets in the conflict graph.

  In this paper, we suggest a novel LP formulation for combinatorial
  auctions with conflict graph using a non-standard graph parameter,
  the so-called \emph{\inductiveindependence{}}. Taking into account
  this parameter enables us to bypass the well-known lower bound of
  $\Omega(n^{1-\varepsilon})$ on the approximability of independent
  set in general graphs with $n$ nodes (bidders). We achieve
  significantly better approximation results by showing that
  interference constraints for wireless networks yield conflict graphs
  with bounded \inductiveindependence{}.

  Our framework covers various established models of wireless
  communication, e.g., the protocol or the physical model. For the
  protocol model, we achieve an $O(\sqrt{k})$-approximation, where $k$
  is the number of available channels. For the more realistic physical
  model, we achieve an $O(\sqrt{k} \log^2 n)$ approximation based on
  edge-weighted conflict graphs. Combining our approach with the
  LP-based framework of Lavi and Swamy, we obtain incentive compatible
  mechanisms for general bidders with arbitrary valuations on bundles
  of channels specified in terms of demand oracles.
\end{abstract}

\section{Introduction}
A major challenge of today's wireless networks and mobile
communication is spectrum management, as devices use common frequency
bands that are subject to interference between multiple transmitters
in the same area. In fact, spectrum allocation has become one of the
key problems that currently limits the growth and evolution of
wireless networks. The reason is that, traditionally, frequencies were
given away to large service providers in a static way by re\-gu\-la\-tors
for entire countries. Examples include FCC auctions in the US or the
auctions for UMTS and LTE that took place in Europe. However, demands
for services vary at different times and in different areas. Depending
on time and place this causes frequency bands licensed for one
application to become overloaded. On the other hand, different bands
are idle at the same time. A promising solution to this problem is to
use market approaches that result in a flexible and thus more
efficient redistribution of access rights -- thereby overcoming the
artificial shortage of available spectrum. In this case, parts of the
spectrum that are currently unused by so-called \emph{primary users}
for the originally intended purpose (such as TV or telecommunication)
can be offered to so-called \emph{secondary users}. Licenses for such
secondary usage are valid only for a local area.

A sustainable approach (concisely termed ``eBay in the Sky''
in~\cite{Zhou2008}) to automatically run such a secondary spectrum
market is to auction licenses for secondary users on a regular
basis. In this paper, we propose a general framework and efficient
algorithms to implement such a secondary spectrum auction. In our
model, there are $n$ secondary users who can bid for bundles of the
$k$ wireless channels. Depending on the scenario a user can correspond
to a base station that strives to cover a specific area or a pair of
devices that want to exchange data (e.g., a base station and a mobile
device). In order to account for channel aggregation capabilities of
modern devices, users should be able to acquire multiple channels.  We
allow each user $v$ to have an arbitrary valuation $b_{v,T}$ for each
subset $T$ of channels. This level of generality is necessary because
of different needs, applications, and hardware abilities of the users,
but also because of different locations, spectrum availability, and
interference conditions. For instance, the presence of a primary user
might allow access to a channel only for a subset of mobile devices
located in selected areas. We assume no restrictions on the valuation
functions, not even monotonicity.

In this paper, we devise approximation algorithms for spectrum
allocation on the secondary market with the objective of maximizing
social welfare. We focus on the underlying combinatorial problems and
describe interference conflicts by an edge-weighted conflict graph. In
unweighted graphs, the vertices represent the bidders and the edges
represent conflicts such that the feasible allocations for a channel
correspond to the independent sets in the conflict graph. For
edge-weighted graphs, we extend the definition of independent set to
weighted edges by requiring the sum of all incoming weights to be less than $1$.
We address the following problem.
\begin{problem}[Combinatorial Auction with Conflict Graph]\label{def:problem}
 Given a graph $G = (V, E)$, a
  natural number $k$, and a valuation function $b\colon V \times
  2^{[k]} \to \NN$, find a feasible allocation $S\colon V \to 2^{[k]}$
  that maximizes the social welfare $b(S) := \sum_{v \in V} b_{v,
    S(v)}$. 

  An allocation $S$ is called feasible if for all channels $j \in
  [k]$, the set of vertices that are assigned to this channel,
  i.\,e. $\{ v \in V \mid j \in S(v) \}$, is an independent set.
\end{problem}
Observe that this problem generalizes combinatorial auctions (where
the conflict graph is a clique) and maximum weight independent set
(where $k=1$). This formulation covers a large number of binary
interference models (such as the protocol model). As we will see, edge
weights allow to express even more realistic models like the physical
model. Here, we can even take the effects of power control into
account.

\subsection{Our contribution}

We devise the first approximation algorithms for the combinatorial
auction problem with conflict graph. Our approach is based on a novel
LP formulation for the independent set problem 
using a non-standard graph parameter.

\begin{definition}[\inductiveindependence{} $\rho$]
\label{def:IIN}
For a graph $G=(V,E)$, the \inductiveindependence{} $\rho$ is the
smallest number such that there is an ordering $\pi$ of the vertices
satisfying: For all $v \in V$ and all independent sets $M \subseteq
V$, we have $\left\lvert M \cap \left\{ u \in V \mid \{u, v\} \in E,
    \pi(u) < \pi(v) \right\} \right\rvert \leq \rho$.
\end{definition}

In words, for every vertex $v \in V$, the size of an independent set
in the {\em backward neighborhood} of $v$, i.e., the set of neighbors
$u$ of $v$ with $\pi(u) < \pi(v)$, is at most $\rho$. Conflict graphs
derived from various simple models of wireless communication with
binary conflicts like, e.g., the protocol model, distance-2 matchings,
or disk graphs, have $\rho = O(1)$, see, e.g.,~\cite{Wan2009}. The
corresponding ordering $\pi$ is efficiently computable in these cases.
We exploit this property in our algorithms.

Our main results concern the so-called physical model which is common
in the engineering community and was only recently subject to
theoretical work. In binary models of wireless communication
usually studied in theoretical computer science, we make the
oversimplifying assumption that interference caused by a signal stops
at some boundary around the sender, and receivers beyond this boundary
are not disturbed by this signal. In contrast, the physical model takes into
account realistic propagation effects and additivity of
signals. Feasibility of simultaneous transmissions is modeled in terms
of signal to interference plus noise ratio (SINR) constraints. We
study two variants of this model, one in which signals are sent at
given powers (e.g., uniform) and one where the powers are subject to
optimization themselves. We show how to represent SINR constraints for
both of these variants in terms of an edge-weighted conflict graph and
introduce appropriate notions of ``independent set'' and ``inductive
independence number'' for edge-weighted graphs. Note that the
combinatorial auctions with edge-weighted conflict graphs can be
defined in the same way as stated in Problem~\ref{def:problem} given
an appropriate definition of ``independent set''.

At first, we prove that the inductive independence number $\rho$ for
edge-weighted graphs obtained from the physical model (in both
variants) is bounded by $O(\log n)$ and the corresponding ordering is
efficiently computable. This enables us to bypass the well-known lower
bound of $\Omega(n^{1-\varepsilon})$ on the approximability of
independent set in general graphs. In particular, we present an LP
relaxation capturing both interference constraints and valuations of
users for subsets of channels. Similar to regular combinatorial
auctions, the LP might require an exponential number of valuations
$b_{v,T}$ to be written down explicitly. However, we show how to solve
the LP using only oracle access to bidder valuations. Our LP based
framework is able to handle edge-weighted conflict graphs resulting
from the physical model. By rounding the LP optimum, our algorithm
achieves an $O(\rho \cdot \sqrt{k} \log n)$ approximation
guarantee. Combining this with the bound on $\rho$ gives an
$O(\sqrt{k} \log^2 n)$-approximation of the social welfare for
spectrum auctions in the physical model (in both variants).

For more simple binary models of wireless communication such as the
protocol model, our approach yields an $\O(\rho \cdot
\sqrt{k})$-approximation. Using the bounds on $\rho$ mentioned above,
this yields an $\O(\sqrt{k})$ approximation guarantee. In this case,
we also provide some complementing hardness results. In general, it is
hard to approximate the combinatorial auction problem with conflict
graphs to a factor of $\O(\rho^{1-\epsilon})$ and to a factor of
$O(k^{\frac 12 - \epsilon})$ for any constant $\epsilon > 0$. While
for some specific models better approximations exist, in general the
bounds provided by our algorithms for binary models cannot be
improved in terms of a single parameter $\rho$ or $k$. In addition,
we provide stronger lower bounds for the case
of asymmetric channels, in which the conflict graph can be different
for each channel. In this case, our algorithm guarantees a factor of
$O(\rho \cdot k)$, which is best possible in general. 

Our approach can be used to derive incentive compatible mechanisms
using the LP-based framework of Lavi and Swamy~\cite{Lavi2005} for
general bidders with demand oracles. In fact, we slightly extend this
framework by starting with an infeasible rather than feasible ILP
formulation. The approximation algorithm computes a linear combination
of feasible solutions approximating the optimal solution of the
corresponding LP and then chooses one of these solutions at
random. The obtained mechanism is truthful in expectation. \\[0.2cm]
{\noindent \bf Outline.} For technical reasons, we present our results in
a different order than stated above. We first introduce the basic
approach in the context of unweighted conflict graphs in
Section~\ref{sec:unweighted}. The extensions to edge-weighted graphs
including formal definitions of independent sets and inductive
independence number are given in Section~\ref{sec:weighted}. The
aforementioned wireless models (especially the variants of the
physical model) are formally introduced in Section~\ref{sec:models},
where we also show the bounds on the inductive independence
number. The application of the framework by Lavi and Swamy is
discussed in Section~\ref{sec:mechanism}. Finally, the results on asymmetric
channels are presented in Section~\ref{sec:asymmetric}.

\subsection{Related Work}
The idea of establishing secondary spectrum markets has attracted much
attention among researchers in applied networking and engineering
communities~\cite{Zhou2008,Ghandhi2007,Buddhikot2005,Ileri2005}.
There are many different fundamental regulatory questions that need to
be addressed when implementing such a market. For example it has to be
clarified who runs the market and who is allowed to sell and buy
spectrum there. Possible actors could be network providers, brokers,
regulators and end-users. In addition, it has to be guaranteed that
existing services are not harmed. In most of the literature on
spectrum markets the technological aspects dominate.  Many results in
this area are only of qualitative nature, only a few examples (such
as~\cite{Zhou2008,Zhou2009}) do explicitly consider truthfulness. We
believe that our combinatorial models based on (edge-weighted)
conflict graphs taking into account the bounded inductive independence
number allows us to neglect technological aspects and to focus on the
underlying combinatorial and algorithmic questions. To the best of our
knowledge there is no previous work on auctions using the general
framework of conflict graphs, or, in general, non-trivial provable
worst-case guarantees on the efficiency of the allocation.

In contrast, combinatorial auctions have been a prominent research
area in algorithmic game theory over the last decade. A variety of
works treats auctions with special valuation functions, such as
submodular valuations or ones expressible by specific bidding
languages. For an introduction see, e.g.,~\cite[Chapters 11 and
12]{Nisan07} or~\cite{Cramton06}. In addition, designing
(non-truthful) approximation algorithms for the allocation problems
has found interest, most notably for submodular valuations
(e.g.,~\cite{Vondrak2008,Feige2006}). More relevant to our work,
however, are results that deal with truthful mechanisms for general
valuations. Most notably, Lavi and Swamy~\cite{Lavi2005} and Dobzinski
et al.~\cite{Dobzinski2006} derive mechanisms using only demand
oracles that achieve an $\sqrt{k}$-approximation with truthfulness in
expectation and universal truthfulness, respectively. A deterministic
truthful $(k/\sqrt{\log k})$-approximation is obtained by Holzman et
al.~\cite{Holzman2004}.

Over the last decades, there has been much research on finding maximum
independent sets in the context of interference models for wireless
networks. One of the simplest models in this area are disk graphs,
which are mostly analyzed using geometric
arguments. See~\cite{Fishkin2003,Graf1994} for a summary on the
results and typical techniques. Recently and independently from our
work, Christodoulou et al.~\cite{CEF10} study combinatorial auctions
for geometric objects. Similar to our approach, they present an LP
formulation based on a property in terms of an ordering, the fatness
of geometric objects.

Akcoglu et al.~\cite{Akcoglu2000} and Ye and Borodin~\cite{Ye2009}
also use the inductive independence number to approximate independent
sets within a factor of $\rho$ with a motivation stemming from chordal graphs.
However, they do not consider multiple channels or wireless communication. 
As the algorithm is not monotone, it is also not immediately applicable 
for truthful auctions.

Algorithmic aspects of the physical model have become popular in
theoretical research recently, particularly the problem of scheduling,
i.e., partitioning a given set of requests in a small number of
classes such that all requests are successful. New challenges arise
since graph-theoretic coloring methods cannot be directly applied. For
example, there have been a number of results on how to choose powers
for short schedule lengths~\cite{Fanghaenel2009, Fanghaenel2009a,
  Halldorsson2009}. A popular method is fixing powers according to
some distance-based scheme. For uniform power assignments, a
constant-factor approximation algorithm for the problem of finding an
independent set (i.e., a maximum set that may share a single channel)
is presented in~\cite{Goussevskaia2009}. An online version of the
problem has been studied in~\cite{Fanghaenel2010} presenting tight
bounds depending on the difference in lengths of the requests. Most
recently, a constant-factor approximation algorithm for arbitrary
power schemes has been obtained by Kesselheim~\cite{Kesselheim2010a}.

\section{Unweighted Conflict Graphs}
\label{sec:unweighted}

\subsection{Our LP relaxation}
\label{sec:lp}

One can get a very intuitive LP formulation for the Weighted
Independent Set problem by leaving out the integer constraints from
the Integer Linear Programm formulation.

\begin{align*}
 \text{Max. } & \sum_{v \in V} b_v x_v\\
 \text{s.\,t. } & x_u + x_v \leq 1 && \text{for all $\{ u, v \} \in E$} \\
 & 0 \leq x_v \leq 1 && \text{for all $v \in V$}
\end{align*}

This LP can be used to approximate Independent Set within a factor of
$(\bar{d} + 1)/2$~\cite{Hochbaum1983,Kako2005} where $\bar{d}$ is the
average vertex degree. However, even for the case of a clique the integrality gap is $\nicefrac{n}{2}$.

In contrast to this edge-based LP formulation, we here present
a different LP based on the \emph{inductive indepence number} $\rho$
(recall Definition~\ref{def:IIN}). As we will see later, in typical
conflict graphs the \inductiveindependence{} is constant and the
corresponding ordering $\pi$ can be efficiently calculated. Here we
use $\Gamma_\pi(v) = \{ u \in V \mid \{u,v\} \in E, \pi(u) < \pi(v)\}$
to denote the backward neighborhood of $v$. This allows to use the
following LP relaxation that has one constraint for each combination of
a vertex and a channel and another one for each vertex.
\begin{subequations}
  \label{eq:lp}
  \begin{align}
   \text{Max. } & \sum_{v \in V} \sum_{T \subseteq [k]} b_{v,T} x_{v,T}\\
    \text{s.\,t. } & \sum_{u \in \Gamma_\pi(v)} \sum_{\substack{T \subseteq [k] \\ j \in T}} x_{u,T} \leq \rho && \text{for all $v \in V$, $j \in [k]$} \label{eq:lp:niceproperty}\\
    & \sum_{T \subseteq [k]} x_{v, T} \leq 1 && \text{for all $v \in V$} \label{eq:lp:sumofsets}\\
    & x_{v, T} \geq 0 && \text{for all $v \in V$, $T \subseteq
      [k]$} \label{eq:lp:nonnegative}
  \end{align}
\end{subequations}
This LP works as follows. For each vertex $v$ and each possible set $T
\subseteq [k]$ of channels assigned to this vertex, there is one
variable $x_{v, T}$. Due to the bounded \inductiveindependence{} all
feasible allocations correspond to solutions of the LP. However, not
all integer solutions of the LP necessarily correspond to feasible
channel allocations. Nevertheless, we will show how to compute a
feasible allocation from each solution.

\begin{lemma}
  \label{lem:feasible}
  Let $S$ be a feasible allocation and $x$ be defined by $x_{v, T} =
  1$ if $S(v) = T$ and $0$ otherwise, then $x$ is a feasible LP
  solution.
\end{lemma}

\begin{proof}
  Conditions~\eqref{eq:lp:sumofsets} and~\eqref{eq:lp:nonnegative} are
  obviously satisfied. Let us now consider
  Condition~\eqref{eq:lp:niceproperty} for some fixed $v \in V$, $j
  \in [k]$. Set $M:=\{ u \in V \mid \pi(u) < \pi(v), j \in S(u)
  \}$. Since $M$ is an independent set, by definition of the
  \inductiveindependence, we have $\lvert M \cap \Gamma_\pi(v) \rvert \leq \rho$.

  On the other hand, we have
  \[
  \sum_{\substack{u \in \Gamma_\pi(v)}} \sum_{\substack{T
      \subseteq [k] \\ j \in T}} x_{u,T} = \lvert M \cap \Gamma_\pi(v)\rvert \leq \rho\enspace.
  \]

  So $x$ is a feasible LP solution. 
\end{proof}

As all coefficients are non-negative, this LP has a packing
structure. In particular, we can observe the following decomposition
property.
\begin{observation}
  Let $x$ be a feasible solution to the LP, and $x^{(1)}$ be a vector
  such that $0 \leq x^{(1)}_{v, T} \leq x_{v,T}$ for all $v \in V$, $T
  \subseteq [k]$. Then $x^{(1)}$ and $x^{(2)} := x - x^{(1)}$ are
  feasible LP solutions as well.
\end{observation}

If there are only $\O(\log n)$ valuations $b_{v, T}$ non-zero, this LP
is solvable in polynomial time. In general, the elementary
representation of the $b_{v, T}$ values is exponential in $k$. We can
still solve the LP optimally if bidders can be represented by demand
oracles.

\subsection{Demand Oracles}
\label{sec:demand}
If there is an arbitrary number of channels, we must define an
appropriate way to query the valuation functions of the requests, as
an elementary description becomes prohibitively large. A standard way
to deal with this issue in the auction literature is the
representation by so-called \emph{demand oracles}. To query the demand
oracle of bidder $v$, we assign each channel $i$ a price $p_i$. Then
the oracle delivers his ``demand'' $S = \arg\max_{T \subseteq [k]}
b_{v,T} - \sum_{i \in T} p_i$, i.\,e., a bundle that maximizes the
utility of $v$ given that he pays the sum of prices of channels in the
bundle. In ordinary combinatorial auctions such demand oracles can be
used to separate the dual of the underlying LP. We here show that such
demand oracles can also be used for the solution of our LP~\eqref{eq:lp}.
Consider the dual given by
\begin{subequations}
\begin{align}
\text{Min. } & \sum_{v \in V} \sum_{j \in [k]} \rho y_{v,j} + \sum_{v \in
V} z_v \\
\text{s.\,t. } & \sum_{\substack{u \in V \\ v \in \Gamma_\pi(u)}} \sum_{j
\in T} y_{u,j} + z_v \geq b_{v,T} && \text{for all $v
\in V$, $T \subseteq [k]$} \\
& y_{v, j} \geq 0 && \text{for all $v \in V$, $j \in T$}
\end{align}
\end{subequations}
In contrast to ordinary combinatorial auctions, we cannot use the
solution $(y,z)$ directly as the channel prices. Instead, we choose
\emph{bidder-specific} channel prices by
\[
p_{v,j} = \sum_{\substack{u \in V \\ v \in \Gamma_\pi(u)}} 
y_{u,j}\enspace.
\]
Using this idea we see that the constraints of the dual are indeed
equivalent to upper bounds on the utility with bidder-specific channel
prices. By obtaining the demand bundle with highest utility for each
player, we find a violated constraint or verify that none exists. This
allows to separate the dual LP and to solve it efficiently using the
ellipsoid method. This way, we get an equivalent primal LP with only
polynomially constraints. The corresponding primal solution has only
polynomially many variables with $x^\ast_{v,T} > 0$.

\subsection{Rounding LP Solutions}

Having described the LP relaxation, we now analyze Algorithm~\ref{alg:unweighted-rounding} computing feasible allocations
from LP solutions as follows. First, it decomposes the given LP
solution to two solutions $x^{(1)}$ and $x^{(2)}$ (line 1). In $x^{(1)}$ all fractional variables $x_{v, T}$ for sets $T$ with $\lvert T \rvert \geq \sqrt{k}$ are set to zero. To get $x^{(2)}$ the exact opposite is performed. From each one, a feasible allocating is computed and the better one is selected at the end. This means, the algorithm either allocates only sets of size at most $\sqrt{k}$ or only of size at least $\sqrt{k}$.
The actual computation of the allocation works the
same way for both LP solutions. It consists of two major parts: a
rounding stage and a conflict-resolution stage. In the rounding stage
(lines 3--4), a tentative allocation is generated as follows. For each
vertex $v$ the set of allocated channels $S^{(l)}(v)$ is determined
independently at random. Each set $T \neq \emptyset$ is taken with
probability $x^{(l)}_{v, T} / 2 \sqrt{k} \rho$ and with the
remaining probability the empty set is allocated.

Conflicts can occur when two adjacent vertices share the same
channel. In this case, the conflict is resolved (lines 5--8) by
allocating the channel to the vertex with smaller index in the $\pi$
ordering. The other vertex is removed from the solution by being allocated
the empty set.

\begin{algorithm}
decompose $x$ into two solutions $x^{(1)}$ and $x^{(2)}$ by $x^{(1)}_{v,T} = x_{v,T}$ if $\lvert T \rvert \leq \sqrt{k}$ and $x^{(1)}_{v,T} = 0$ otherwise. $x^{(2)} = x - x^{(1)}$ \;
\For{$l \in \{1, 2\}$}{
\For(\tcc*[f]{Rounding Stage}){$v \in V$}{
with probability $\frac{x^{(l)}_{v, T}}{2 \sqrt{k} \rho}$ set $S^{(l)}(v) := T$\;
}
\For(\tcc*[f]{Conflict-Resolution Stage}){$v \in V$}{
\For{$u \in V$ with $\pi(u) < \pi(v)$ and $\{ u, v\} \in E$}{
\If{$S^{(l)}(u) \cap S^{(l)}(v) \neq \emptyset$}{
$S^{(l)}(v) := \emptyset$\;
}
}
}
}
return the better one of the solutions $S^{(1)}$ and $S^{(2)}$\;
\caption{LP rounding algorithm for the combinatorial auction problem
  with unweighted conflict graphs}
\label{alg:unweighted-rounding}
\end{algorithm}

\begin{theorem}
  For any feasible LP solution $x^\ast$ with value $b^\ast$,
  Algorithm~\ref{alg:unweighted-rounding} calculates a feasible
  allocation $S$ of value at least $\nicefrac{b^\ast}{8 \sqrt{k}
    \rho}$ in expectation.
\end{theorem}

\begin{proof}
  The allocations $S^{(1)}$ and $S^{(2)}$ are obviously feasible
  allocations because if $\{u, v\} \in E$, then $S^{(1)}(u) \cap
  S^{(1)}(v) = \emptyset$ and $S^{(2)}(u) \cap S^{(2)}(v) =
  \emptyset$. Therefore, the output is also a feasible allocation.

  Let us now bound the expected values of solutions $S^{(1)}$ and
  $S^{(2)}$. Let $l \in \{ 1, 2\}$ be fixed. Let $X_{v, T}$ be a 0/1
  random variable indicating if $S^{(l)}(v)$ is set to $T$ after the
  rounding stage. Clearly, we have

  \begin{equation}
    \Ex{X_{v, T}} = \frac{x^{(l)}_{v, T}}{2 \sqrt{k} \rho} \enspace. \label{eq:unweighted-analysis:firstprobability}
  \end{equation}

  Let $X'_{v, T}$ be a 0/1 random variable indicating if $S^{(l)}(v)$
  is set to $T$ after the conflict-resolution stage. We consider the
  event that $X'_{v, T} = 0$, given that $X_{v, T} = 1$, i.\,e. that
  $v$ is removed in the conflict-resolution stage after having
  survived the rounding stage.

  \begin{lemma}
    \label{lemma:unweighted-analysis:lostinconflictresolution}
    The probability of being removed in the conflict-resolution stage
    after having survived the rounding stage is at most
    $\nicefrac{1}{2}$.
  \end{lemma}

  \begin{proof}
    The event can only occur if $X_{u, T'} = 1$ for some $u \in V$
    with $\pi(u) < \pi(v)$, $\{ u, v \} \in E$, and $T \cap T' \neq
    \emptyset$. In terms of the random variables $X_{u, T}$ this is
    \[
    \sum_{u \in \Gamma_\pi(v)} \sum_{\substack{T'
        \subseteq [k]\\ T \cap T' \neq \emptyset}} 
    X_{u, T'} \geq 1 \enspace.
    \]

    Using this notation we can bound the probability of the event by
    using the Markov inequality
    \begin{align*}
      \Pr{ X'_{v, T} = 0 \mid X_{v, T} = 1} & \leq \Pr{ \sum_{u \in \Gamma_\pi(v)} \sum_{\substack{T' \subseteq [k]\\ T \cap T' \neq \emptyset}} X_{u, T'} \geq 1 } \\
      & \leq \Ex{ \sum_{u \in \Gamma_\pi(v)}
        \sum_{\substack{T' \subseteq [k]\\ T \cap T' \neq \emptyset}} X_{u, T'} } \enspace.
    \end{align*}
 
    We will now show separately that this expectation is at most
    $\nicefrac{1}{2}$ for each of the two possible values of $l$
    ($l=1$ or $l=2$).

    \begin{description}
    \item[Case 1 ($l = 1$):] We have:
\[
\Ex{ \sum_{u \in \Gamma_\pi(v)} \sum_{\substack{T'
      \subseteq [k]\\ T \cap T' \neq \emptyset}} X_{u, T'} } \leq \Ex{ \sum_{j \in T} \sum_{u \in \Gamma_\pi(v)} \sum_{\substack{T' \subseteq [k]\\ j \in T'}}
  X_{u, T'} } \enspace.
\]

Due to linearity of expectation this is equal to
\[
\sum_{j \in T} \sum_{u \in \Gamma_\pi(v)} \sum_{\substack{T' \subseteq [k] \\ j \in T'}} \Ex{ X_{u, T'} } \enspace.
\]

Using Equation~\eqref{eq:unweighted-analysis:firstprobability} and the
fact that $x^{(1)}$ is an LP solution, this is
\[
\sum_{j \in T} \sum_{u \in \Gamma_\pi(v)} \sum_{\substack{T' \subseteq [k] \\ j \in T'}} \frac{x^{(1)}_{u, T'}}{2 \sqrt{k} \rho} \leq \sum_{j \in T}\frac{1}{2 \sqrt{k}} \enspace.
\]

Recall that we only have to deal with sets $T$ for which $\lvert T
\rvert \leq \sqrt{k}$ in this case. Hence, the expectation is at most
$\nicefrac{1}{2}$, and so is the probability that $v$ is removed in
the conflict-resolution stage.

\item[Case 2 ($l = 2$):] In this case, we have $X_{u, T'} > 0$ only for sets $T'$ with $\lvert T' \rvert \geq \sqrt{k}$. This yields for all $u \in V$
\begin{align*}
\sum_{\substack{T' \subseteq [k] \\ T \cap T' \neq \emptyset}} X_{u, T'} \leq \sum_{\substack{T' \subseteq [k] \\ T' \neq \emptyset}} X_{u, T'} 
 = \sum_{\substack{T' \subseteq [k]\\ T' \neq \emptyset}} \sum_{j \in T'} \frac{X_{u, T'}}{\lvert T' \rvert} 
 = \sum_{j \in [k]} \sum_{\substack{T' \subseteq [k]\\ j \in T'}} \frac{X_{u, T'}}{\lvert T' \rvert} 
 \leq \frac{1}{\sqrt{k}} \sum_{j \in [k]} \sum_{\substack{T' \subseteq [k]\\ j \in T'}} X_{u, T'} \enspace.
\end{align*}
So, we get
\[
\Ex{ \sum_{u \in \Gamma_\pi(v)} \sum_{\substack{T' \subseteq [k]\\ T \cap T' \neq \emptyset}} X_{u, T'} } \hspace{-.05cm} \leq \Ex{ \frac{1}{\sqrt{k}} \sum_{j \in [k]} \sum_{u \in \Gamma_\pi(v)} \sum_{\substack{T' \subseteq [k]\\ j \in T'}} X_{u, T'} } \enspace.
\]

Again, we use linearity of expectation,
Equation~\eqref{eq:unweighted-analysis:firstprobability} and the fact
that $x^{(2)}$ is an LP solution. This gives us
\[
\frac{1}{\sqrt{k}} \sum_{j \in [k]} \sum_{u \in \Gamma_\pi(v)} \sum_{\substack{T' \subseteq [k] \\ j \in T'}} \frac{x^{(2)}_{u, T'}}{2 \sqrt{k} \rho} \leq \frac{1}{\sqrt{k}} \sum_{j \in [k]}\frac{1}{2 \sqrt{k}} \leq \frac{1}{2} \enspace.
\]
This bounds the probability for the second case.
\end{description}

In both cases we have $\Pr{ X'_{v, T} = 0 \mid X_{v, T} = 1 } \leq \nicefrac{1}{2}$.
\end{proof}
Using Lemma~\ref{lemma:unweighted-analysis:lostinconflictresolution}
and Equation~\eqref{eq:unweighted-analysis:firstprobability} we get
for all $v \in V$ and $T \subseteq [k]$
\[
\Ex{X'_{v, T}} \geq \frac{x^{(l)}_{v, T}}{4 \sqrt{k} \rho} \enspace.
\]
This yields that both calculated solutions $S^{(l)}$ for $l \in \{ 1,
2 \}$ have expected value
\begin{align*}
\Ex{ b(S^{(l)}) } = \Ex{ \sum_{v \in V} \sum_{T \subseteq [k]} b_{v, T} \cdot X'_{v, T}} = \sum_{v \in V} \sum_{T \subseteq [k]} b_{v, T} \cdot \Ex{X'_{v, T}} \geq \frac{1}{4 \sqrt{k} \rho} \sum_{v \in V} \sum_{T \subseteq [k]} b_{v, T} x^{(l)}_{v, T} \enspace.
\end{align*}
So, the better ones of the two solutions has expected value
\begin{flalign*}
\Ex{ \max \{ b(S^{(1)}), b(S^{(2)}) \} } & \geq \frac{1}{2} \left( \Ex{ b(S^{(1)})} + \Ex{b(S^{(2)}) } \right) \\
& \geq \frac{1}{8 \sqrt{k} \rho} \sum_{v \in V} \sum_{S \subseteq [k]} b_{v, S} \left( x^{(1)}_{v, S} + x^{(2)}_{v, S} \right) \\
& = \frac{1}{8 \sqrt{k} \rho} \sum_{v \in V} \sum_{S \subseteq [k]} b_{v, S} x_{v, S} = \frac{b^\ast}{8 \sqrt{k} \rho}\enspace.  
\end{flalign*}
\end{proof}


\subsection{Hardness Results}
\label{sec:hardness}

In this section we provide matching lower bounds for the approximation
ratios of our algorithms. This shows that the above results cannot be
vitally improved without further restricting the model. Our results
are based on the hardness of approximating independent set in
bounded-degree graphs~\cite{Trevisan2001} or general
graphs~\cite{Hastad1999}.
A first result is that the $\O(\rho)$ algorithm for the case $k=1$ is
almost optimal.

\begin{theorem}
  For $k=1$ and for each $\rho=\O(\log n)$ there is no
  $\nicefrac{\rho}{2^{\O(\sqrt{\log \rho})}}$ approximation algorithm
  unless \classP = \classNP.
\end{theorem}

\begin{proof}
  Such an algorithm could be used to approximate Independent Set in
  bounded-degree graphs. Given a graph with maximum degree $d$ its
  \inductiveindependence{} $\rho$ is also at most
  $d$. Trevisan~\cite{Trevisan2001} shows that there is no
  $\nicefrac{d}{2^{\O(\sqrt{\log d})}}$-approximation algorithm for
  all $d=\O(\log n)$ unless \classP = \classNP. This directly yields
  the claim.
\end{proof}

As a second result we can also prove the impact of the number of
channels $k$ has to be as large as $\sqrt{k}$.
 
\begin{theorem}
  Even for $\rho = 1$ there is no
  $k^{\frac{1}{2}-\varepsilon}$-approximation algorithm unless
  \classZPP = \classNP.
\end{theorem}

Our framework extends general combinatorial auctions with $k$ items,
and this is a standard result in the area~\cite[Chapter 9]{Nisan07}
derived from the hardness of independent set in general graphs.

In conclusion, our algorithmic results are supported by almost
matching lower bounds in each parameter. Without further restricting the
graph properties (which means to use additional properties of an
interference model) no vitally better approximation guarantees can be
achieved in terms of $\rho$ resp. $k$. However, this does not prove no
$O(\rho + \sqrt{k})$ approximation can exist.

\section{Edge-weighted Conflict Graphs}
\label{sec:weighted}

In this section we extend conflicts over binary relations
(conflict/no-conflict). In wireless communication, we encounter
situations that a radio transmission is exposed to interference by a
number of devices relatively far away. If there was only a single one
of them, interference would be acceptable but their overall
interference is too high. For such aggregation aspects we introduce
edge-weighted conflict graphs, in which there is a non-negative weight
$w(u,v)$ between any pair of vertices $u, v \in V$. An
\emph{independent set} is defined as a set $M \subseteq V$ such that
$\sum_{u \in M} w(u, v) < 1$ for all $v \in M$.

The definition of the \inductiveindependence{} can be generalized in a
straightforward way. Since edge weights need not be symmetric, it
turns out to be convenient to use the following symmetric edge weights
$\bar{w}(u, v) = w(u, v) + w(v, u)$.
\begin{definition}
  The \emph{\inductiveindependence} of an edge-weighted graph $G$ is
  the minimum number $\rho$ such that there is a total ordering
  $\pi\colon V \to [n]$ (bijective function) which fulfills for all
  vertices $v$ and all independent sets $M \subseteq \left\{ u \in V
    \mid \pi(u) < \pi(v) \right\}$ the following condition:
  \[
  \sum_{u \in M} \bar{w}(u, v) \leq \rho \enspace.
  \]
\end{definition}

In the same way as in the unweighted case, we can use the definition
to formulate the LP relaxation.
\begin{subequations}
  \label{eq:lp-weighted}
  \begin{align}
    \text{Max. }& \sum_{v \in V} \sum_{T \subseteq [k]} b_{v,T} x_{v,T} && \\
    \text{s.\,t. } & \sum_{\substack{u \in V \\ \pi(u) < \pi(v)}} \sum_{\substack{T \subseteq [k] \\ j \in T}} \bar{w}(u, v) \cdot x_{u,T} \leq \rho  && \text{for all $v \in V$, $j \in [k]$} \label{eq:lp-weighted:niceproperty} \\
   & \sum_{T \subseteq [k]} x_{v, T} \leq 1 && \text{for all $v \in V$} \\
   & x_{v, T} \geq 0 && \text{for all $v \in V$, $T \subseteq
      [k]$} 
  \end{align}
\end{subequations}

In weighted graphs we lose an important property we made use of in
unweighted graphs: Resolving conflicts in one direction
only does not suffice. To cope with this issue, we increase the
scaling by another factor of 2. We use rounding and conflict
resolution as previously to ensure that for each vertex $v$ the sum of
edge weights to neighboring vertices that have smaller indices and
share a channel with $v$ is at most $\nicefrac{1}{2}$. Formally, a
\emph{partly-feasible allocation} is an allocation $S\colon V \to
2^{[k]}$ such that
\begin{equation}
  \sum_{\substack{u \in V \\ \pi(u) < \pi(v) \\ S(v) \cap S(u) \neq \emptyset}} \bar{w}(u,v) < \frac{1}{2} \qquad \text{for all $v \in V$} \enspace. \label{eq:weighted-partlyfeasible}
\end{equation}
Rounding LP solutions to such partly-feasible allocations can be
carried out in a similar way as Algorithm~\ref{alg:unweighted-rounding}.
Algorithm~\ref{alg:weighted-rounding} decomposes the given LP solution
the same way as Algorithm~\ref{alg:unweighted-rounding}. Afterwards, it also performs two
stages. In the rounding stage (lines~2--4), again a tentative
allocation is determined randomly by considering the LP solution as a
probability distribution.

Afterwards, only a partial conflict resolution (lines~5--8) is
performed: If for some vertex $v$ the sum of edge weights to neighbors
that have lower $\pi$ values and share a channel exceeds
$\nicefrac{1}{2}$, it is removed from the solution (i.\,e. it is
allocated the empty set). Such a partly-feasible solution satisfies
Equation~\eqref{eq:weighted-partlyfeasible}.
 
\begin{algorithm}
decompose $x$ into two solutions $x^{(1)}$ and $x^{(2)}$ by $x^{(1)}_{v,T} = x_{v,T}$ if $\lvert T \rvert \leq \sqrt{k}$ and $x^{(1)}_{v,S} = 0$ otherwise. $x^{(2)} = x - x^{(1)}$\;
\For{$l \in \{1, 2\}$}{
\For(\tcc*[f]{Rounding Stage}){$v \in V$}{
with probability $\frac{x_{v, T}}{4 \sqrt{k} \rho}$ set $S^{(l)}(v) := T$}
\For(\tcc*[f]{Partial Conflict-Resolution Stage}){$v \in V$}{
set $U(v) := \{ u \in V \mid \pi(u) < \pi(v),\linebreak S^{(l)}(v) \cap S^{(l)}(u) \neq \emptyset \}$\;
\If{$\sum_{u \in U(v)} \bar{w}(u, v) \geq \frac{1}{2}$}{
$S^{(l)}(v) := \emptyset$
}
}
}
return the better one of the allocations $S^{(1)}$ and $S^{(2)}$
\caption{LP rounding algorithm for the combinatorial auction problem
  with weighted conflict graphs}
\label{alg:weighted-rounding}
\end{algorithm}

\begin{lemma}
  For any feasible LP solution $x^\ast$ with value $b^\ast$,
  Algorithm~\ref{alg:weighted-rounding} calculates a partly-feasible
  allocation $S$ of value at least $\nicefrac{b^\ast}{16 \sqrt{k}
    \rho}$ in expectation.
\end{lemma}

\begin{proof}
  The allocation is partly feasible since both allocations $S^{(1)}$
  and $S^{(2)}$ satisfy Condition~\eqref{eq:weighted-partlyfeasible}.

  For the value of the solution let us again bound the value of the
  partly-feasible allocations $S^{(1)}$ and $S^{(2)}$. Again, let us
  fix $l \in \{ 1, 2\}$. Let $X_{v, T}$ be a 0/1 random variable
  indicating if $S^{(l)}(v)$ is set to $T$ after the rounding
  stage. This time, we have
  \begin{equation}
    \Ex{X_{v, T}} = \frac{x^{(l)}_{v, T}}{4 \sqrt{k} \rho} \enspace. \label{eq:weighted-analysis:firstprobability}
  \end{equation}

  Let $X'_{v, T}$ be a 0/1 random variable indicating if $S^{(l)}(v)$
  is set to $T$ after the partial conflict-resolution stage. Again, we
  consider the event that $X'_{v, T} = 0$, given that $X'_{v, T} = 1$,
  i.\,e., that $v$ is removed in the conflict-resolution stage after
  having survived the rounding stage.

  Lemma~\ref{lemma:unweighted-analysis:lostinconflictresolution}
  cannot be directly applied in this case. However, we have
  \begin{align*}
    \Pr{ X'_{v, T} = 0 \mid X_{v, T} = 1 } & \leq \Pr{ \sum_{\substack{u \in V \\ \pi(u) < \pi(v)}} \sum_{\substack{T' \subseteq [k] \\ T \cap T' \neq \emptyset}} \bar{w}(u, v) \cdot X_{u, T'} \geq \frac{1}{2} } \\
    & \leq 2 \Ex{ \sum_{\substack{u \in V \\ \pi(u) < \pi(v)}} \sum_{\substack{T' \subseteq [k] \\ T \cap T' \neq \emptyset}} \bar{w}(u, v) \cdot X_{u, T'}  } \\
    & = \Ex{ \sum_{\substack{u \in V \\ \pi(u) < \pi(v)}}
      \sum_{\substack{T' \subseteq [k] \\ T \cap T' \neq \emptyset}}
      \bar{w}(u, v) \cdot 2 X_{u, T' } }
  \end{align*}
  due to the Markov inequality and linearity of expectation. Now, we
  can use exactly the same arguments as in the proof of
  Lemma~\ref{lemma:unweighted-analysis:lostinconflictresolution}
  literally. This proof relies on two conditions: the bound on
  $\Ex{X_{u, T}}$ and the fact that $x^{(l)}$ is a feasible LP
  solution. Both conditions are again satisfied.

  This implies that $\Pr{ X'_{v, T} = 0 \mid X_{v, T} = 1 } \leq
  \nicefrac{1}{2}$ for both cases $l \in \{ 1, 2 \}$. In combination
  with Equation~\eqref{eq:weighted-analysis:firstprobability}, we get
  for all $v \in V$, $T \subseteq [k]$
  \[
  \Pr{ X'_{v, T} = 1 } \geq \frac{x^{(l)}_{v, T}}{8 \sqrt{k} \rho}
  \]
  Thus, we can conclude that for $l \in \{ 1, 2\}$, we have
  \[
  \Ex{b(S^{(l)})} \geq \frac{1}{8 \sqrt{k} \rho} \sum_{v \in V}
  \sum_{T \subseteq [k]} b_{v, T} x^{(l)}_{v, T} \enspace.
  \]

  The expected value of the output is at least
  \[
  \Ex{ \max \{ b(S^{(1)}), b(S^{(2)}) \} } \geq
  \frac{1}{16 \sqrt{k} \rho} \sum_{v \in V} \sum_{T \subseteq
    [k]} b_{v, T} \left( x^{(1)}_{v, T} + x^{(2)}_{v, T} \right) =
  \frac{b^\ast}{16 \sqrt{k} \rho} \enspace.
  \]
\end{proof}

Given a partly-feasible allocation $S$,
Algorithm~\ref{alg:weighted-makefeasible} implements the necessary additional
conflict resolution to derive a fully-feasible one. The algorithm decomposes the partly-feasible
allocation to a number of feasible candidate allocations $S_1$, $S_2$,
\ldots. Each allocation $S_i$ is initialized such that $S_i(v) = S(v)$
if vertex $v$ has been removed from all previous allocations $S_1,
\ldots, S_{i-1}$. Otherwise $S_i(v) = \emptyset$. Then a conflict
resolution is performed on $S_i$: The vertices are considered by
decreasing indices in the $\pi$ ordering. If the weight bound is
violated for some vertex $v$ in the current allocation $S_i$, it is
removed from the allocation by allocating the empty set. At the end,
the best one of the candidate allocations is returned. We will see
that each candidate allocation allocates at least half of the remaining
vertices a non-empty set. Therefore at most $\lceil \log n \rceil$ candidates are computed and the
best one has value at least $\nicefrac{b(S)}{\lceil \log n \rceil}$.
\begin{algorithm}
$i := 1$ \; $V' := V$ \;
\While{$V' \neq \emptyset$}{
initialize $S_i$ by $S_i(v) := S(v)$ for $v \in V'$ and $S_i(v) := \emptyset$ otherwise \;
\For{$v \in V'$ in order of decreasing $\pi$ values}{
\If{$\sum_{\substack{u \in V',\; S_i(v) \cap S_i(u) \neq \emptyset}} \bar{w}(u,v) < 1$}{
delete $v$ from $V'$ \tcc*[f]{$v$ stays in $S_i$}\;
}
\Else{
$S_i(v) := \emptyset$ \tcc*[f]{$v$ is removed from $S_i$}\;
}
}
$i := i + 1$\;
}
return the best one of the allocations $S_1$, $S_2$, \ldots
\caption{Making a partly-feasible allocation fully feasible}
\label{alg:weighted-makefeasible}
\end{algorithm}
\begin{lemma}
\label{lemma:weighted-makefeasible} 
Given a (not necessarily feasible) allocation $S$ in which
Condition~\ref{eq:weighted-partlyfeasible} is fulfilled for all $v \in
V$, Algorithm~\ref{alg:weighted-makefeasible} calculates a feasible
allocation of value at least $\nicefrac{b(S)}{\lceil \log n \rceil}$.
\end{lemma}
\begin{proof}
Obviously, by construction all candidates are feasible and so is the output allocation.

Next, we prove that we need at most $\log n$ iterations of the while loop by showing that in each iteration at most half of the remaining vertices are removed from the allocation. This means the cardinality of $V'$ is at least halved in each iteration. Let $V'_i$ be the set $V'$ after the $i$th iteration of the \emph{while} loop; $V'_0 = V$. 

Let us fix $i \in \mathbb{N}$, and $v \in V'_{i+1}$. We know that $v$ has been removed from $S_i$ by the algorithm. This only happens if
\[
\sum_{\substack{u \in V'_i \\ S_i'(v) \cap S_i'(u) \neq \emptyset}} \bar{w}(u,v) \geq 1
\]
where $S_i'$ is the current state of $S_i$ while the algorithm considers $v$.
Since Equation~\ref{eq:weighted-partlyfeasible} is obviously also satisfied for $S_i'$, it has to be
\[
\sum_{\substack{u \in V'_i \\ \pi(u) > \pi(v) \\ S_i'(v) \cap S_i'(u) \neq \emptyset}} \bar{w}(u,v) \geq \frac{1}{2} \enspace.
\]

For a vertex $u \in V'_i$ with $\pi(u) > \pi(v)$ we know that the vertex has either been removed from the allocation before (then $u \in V'_{i+1}$) or it stays in $S_i$ (i.\,e. $S'_i(u) = S_i(u) = S(u)$ and $u \notin V'_{i+1}$). Hence 
\[
S_i'(u) = \begin{cases}
          \emptyset & \text{ if $u \in V'_{i+1}$} \\
          S(u) & \text{ else} \\
          \end{cases} \enspace.
\]

Combining these two insights, we get a necessary condition: if $v \in V'_{i+1}$ then
\[
\sum_{u \in U_i(v) \setminus U_{i+1}(v)} \bar{w}(u,v) \geq \frac{1}{2} \enspace,
\]
where $U_i(v) = \left\{u \in V'_i \mid \pi(u) > \pi(v), S(v) \cap S(u) \neq \emptyset \right\}$. Summing up all $v \in V'_{i+1}$ we get
\[
\sum_{v \in V'_{i+1}} \sum_{ u \in U_i(v) \setminus U_{i+1}(v)} \bar{w}(u,v) \geq \frac{1}{2} \lvert V'_{i+1} \rvert \enspace.
\]

On the other hand, we can change the ordering of the sums and use the symmetry of the weights $\bar{w}$ to get
\begin{align*}
\sum_{v \in V'_{i+1}} \sum_{u \in U_i(v) \setminus U_{i+1}(v)} \bar{w}(u,v)
& = \sum_{u \in V'_i  \setminus V'_{i+1}} \sum_{\substack{v \in V'_{i+1}\\ \pi(u) > \pi(v) \\ S(v) \cap S(u) \neq \emptyset}} \bar{w}(u,v) \\
& = \sum_{v \in V'_i  \setminus V'_{i+1}} \sum_{\substack{u \in V'_{i+1}\\ \pi(u) < \pi(v) \\ S(v) \cap S(u) \neq \emptyset}} \bar{w}(u,v) \\
& < \frac{1}{2} \lvert V'_i  \setminus V'_{i+1} \rvert
\end{align*}
where the last bound is due to Condition~\eqref{eq:weighted-partlyfeasible}.

In combination this yields
\[
\lvert V'_{i+1} \rvert < \lvert V'_i \setminus V'_{i+1} \rvert \enspace,
\]
which implies
\[
\lvert V'_{i+1} \rvert < \frac{1}{2} \lvert V'_i\rvert \enspace,
\]
meaning less than half of the remaining vertices are removed in each iteration. 

So, since $\lvert V'_0 \rvert = n$, we can conclude that
\[
\lvert V'_i \rvert < \frac{1}{2^i} \cdot n \enspace.
\]
Therefore, we get $\lvert V'_{\lceil \log n \rceil} \rvert < 1$. Thus the algorithm terminates within $\lceil \log n \rceil$ steps.

By definition, for all vertices $S_i(v) = S(v)$ for exactly one $i \in [\lceil \log n \rceil]$ and $S_i(v) = \emptyset$ else. So $\sum_{i \in [\lceil \log n \rceil]} b(S_i) = b(S)$. This yields for the value of the output
\[
\max_{i \in [\lceil \log n \rceil]} b(S_i) \geq \frac{1}{\lceil \log n \rceil} \sum_{i \in [\lceil \log n \rceil]} b(S_i) = \frac{b(S)}{\lceil \log n \rceil} \enspace.
\]
\end{proof}
As a consequence, the computed feasible allocation has a value that in
expectation is at most an $\O(\sqrt{k} \rho \log n)$ factor smaller
than that of the optimal LP solution.
\section{Applications}
\label{sec:models}

In the previous sections we have described a general algorithmic
approach to channel allocation problems when the underlying conflict
graph has bounded \inductiveindependence{}. Here we will show that
this property is particularly wide-spread among models for
interference in wireless communication. Our aim is not to prove optimal bounds in each case but to
show why we believe a bounded \inductiveindependence{} to be a key
insight for understanding algorithmic problems in wireless networking.

The concept of conflict graphs
can be applied in two basic scenarios. On the one hand, the task could
be to allocate channels to \emph{transmitters}. Each transmitter
intends to cover a certain area, e.g., a base station in a cellular
network. The interference model defines which transmitters can be
assigned the same channels. On the other hand, instead of single
transmitters one can consider pairs of network nodes (\emph{links})
that act as sender and receiver. In such a scenario, ``users'' are no
single network nodes but links. Therefore, the vertices of the
conflict graph are links, and edges define which links can be assigned
the same channels.

\subsection{Transmitter Scenarios}
A very simple, yet instructive model for a transmitter scenario is as
follows. We have $n$ transmitters located in the plane at points $p_1,
\ldots, p_n \in \RR^2$. Each of the transmitters has a transmission
range $r_1, \ldots, r_n \in \RR_{>0}$. Transmitters may be assigned
the same channel if their transmission ranges do not intersect. Under
these conditions interference constraints can be modeled by a disk
graph. There is an edge between two vertices if the transmission-range
disks around the corresponding receivers intersect.

\begin{proposition}
  \label{prop:diskgraph}
  Disk graphs have an \inductiveindependence{} of $\rho \leq 5$.
\end{proposition}

\begin{proof}
  Let $G=(V,E)$ be a disk graph. Let $\pi$ be the ordering of vertices
  by decreasing radius of the corresponding disk. So in other words $V
  = \{ v_1, \ldots, v_n \}$ with $r_1 \geq r_2 \geq \ldots \geq r_n$,
  where $r_i$ is the radius of the disk around $v_i$. If the disk
  representation is given, this ordering can be computed in polynomial
  time by simply sorting the vertices. Let be $v \in V$ and $M
  \subseteq \{ u \in V \mid \pi(u) < \pi(v) \}$ be an independent set
  in $G$. By definition of the ordering, we have for all $u \in M$ the
  radius is at least $r_{\pi(v)}$.

  In order to show $\lvert M \rvert \leq 5$, we assume $\lvert M
  \rvert \geq 6$. This would yield that there were two vertices whose
  angle seen from $v$ was at most $60^\circ$. From simple geometric
  arguments we can conclude there has to be an edge between these two
  vertices. This contradicts the assumption that $M$ is an independent
  set and thereby proves the claim.
\end{proof} 

Another example for the transmitter scenario is the so-called
\emph{distance-2 coloring}. In contrast to the above model not only
the neighbors (with intersecting disks) must be on different channels
but also their neighbors. Distance-2 coloring is a common model of
transmitter scenarios. Here, we analyze the restriction on two graph
classes. We refer the reader to \cite{Krumke2001} for the exact
definitions and a discussion of the model. We can prove $\rho = \O(1)$
as well in this case.

\begin{lemma}
  Let $r > 0$, $a>0$ and $D$ be a disk of radius $ar$. Then the number
  of disks of radius at least $r$ that intersect $D$ but not each
  other is at most $(a+2)^2$.
\end{lemma}

\begin{proof}
  W.\,l.\,o.\,g., we assume the surrounding disks to have radius exact
  $r$. By scaling them down and moving them inside their original
  area, they still do not intersect each other. By moving them to the
  respective closest location to $D$, they still intersect $D$.

  The disks of radius $r$ are fully contained within the disk of
  radius $kr + 2r$ around the center of $D$. Each takes an area of
  $\pi r^2$, whereas the available area is only $\pi (ar + 2r)^2$. So,
  the number of surrounding disks is at most $\pi (ar + 2r)^2 / \pi
  r^2 = (a+2)^2$.
\end{proof}

\begin{proposition}
  For Distance-2 coloring in disk graphs the associated conflict graph
  has an \inductiveindependence{} $\rho = \O(1)$.
\end{proposition}

\begin{proof}
  As for disk graphs, we order the vertices by decreasing ranges. Now
  consider a vertex $v$ and a conflicting vertex $u$ of larger
  range. This vertex can either be directly connected to $v$ (there
  are at most $5$ ones of this kind) or via an intermediate vertex
  $u'$. If the $r_{u'} < r_v$ is smaller, we see that the disk of
  radius $r_u$ around $u$ intersects the one of radius $2 r_v$ around
  $v$. The above lemma yields that there can be at most a constant
  number of such vertices. For the case $r_{u'} \geq r_v$, we take
  into consideration the disks around the intermediate vertices also
  do not intersect. So, there can be at most $5$ intermediate vertices
  and as many conflicting vertices. The total number of conflicting
  vertices is constant.
\end{proof}

\begin{proposition}
  For Distance-2 coloring in $(r,s)$-civilized graphs the
  \inductiveindependence{} of the associated conflict graph is $\rho
  \leq (4r/s + 2)$.
\end{proposition}

\begin{proof}
  In this case, the ordering does not matter. Therefore, we do not
  need to know the geometric representation of the graph.

  Consider a vertex $v$ and a set of vertices $M$ conflicting with $v$
  but not with each other. Since the path length from $v$ to each
  vertex in $M$ is at most $2$, the distance in the plane is at most
  $2r$. Now consider disks around the vertices in $M$, each of radius
  $\nicefrac{s}{2}$. By definition of the $(r,s)$-civilized graph
  these disks do not intersect each other. However, each of them
  intersects a disk of radius $2r$ around $v$. Applying the above
  lemma, we see there are at most $(4r/s + 2)^2$ such disks.
\end{proof}

As a matter of fact $\rho$ has to depend on this ratio of $r$ and
$s$. Obviously, all graphs can be represented as $(r, s)$-civilized if
the ratio $r$ and $s$ is unbounded. However, our algorithm's running
time does not depend on $r$ and $s$. Therefore, the approximation
factor has to depend on them.

\subsection{Unweighted Link-Based Scenarios}

There are a number of different interference models for link-based
scenarios that can be described by some unweighted conflict
graph. They are often called graph-based interference models, but to
avoid ambiguities we refer to them as \emph{binary interference
  models}. Due to the large variety, we have to confine ourselves to
some selected examples.
 
Probably the best known binary model is the \emph{Protocol
  Model}~\cite{Gupta2000}. Network nodes are modeled by points located
in the plane. A link consisting of sender $s$ and receiver $r$ may be
allocated to a channel if and only if for all other senders $s'$ on
this channel $d(s', r) \geq (1 + \Delta) d(s, r)$ for some constant
$\Delta > 0$.
\begin{proposition}[Wan~\cite{Wan2009}]
  For the protocol model, the resulting conflict graph has an
  \inductiveindependence{} of
  \[
  \rho \leq \left\lceil \pi / \arcsin \frac{\Delta}{2(\Delta + 1)}
  \right\rceil -1 \enspace.
  \]
\end{proposition}
The \emph{IEEE 802.11 Model} by Alicherry et al.~\cite{Alicherry2005}
is a bidirectional variant of this model, and in this case $\rho \leq
23$~\cite{Wan2009}.
 
A more graph-theoretical approach is \emph{distance-2
  matching}~\cite{Balakrishnan2004}. In this case, two edges $e \neq
e'$ may be allocated to the same channel if there are at least two
edges on any connecting path. Typically, results are restricted to
certain graph classes, because in general approximating maximum
distance-2 matchings is hard. For disk graphs, we can also show that
the corresponding conflict graph has $\rho = \O(1)$. Interestingly,
for distance-2 matching there is already an algorithm and analysis
based on the observation that the \inductiveindependence{} is bounded,
but the concepts are termed differently. Barrett et
al. \cite{Barrett2006} analyze a greedy approach to find a maximum
independent set. For a link $e = (u, v)$, they define $r(e) = r(u) +
r(v)$, where $r(u)$ and $r(v)$ are the radius of the disk surrounding
$u$ resp. $v$. The algorithm orders the links by increasing values of
$r(e)$. The key observation is now that for all links $e$ the maximum
number of links of higher index that collide with $e$ but not with
each other is $\O(1)$. This immediately yields $\rho = \O(1)$.

\begin{corollary}
  For distance-2 matching in disk graphs the associated conflict graph
  has an \inductiveindependence{} $\rho = \O(1)$.
\end{corollary}

Analyses of greedy algorithms are often carried out in a similar
manner. Such arguments already suffice to bound the
\inductiveindependence{}. There is plenty of opportunity to further
extend our results by similar observations.

\subsection{Physical Model}

The models mentioned above go well with graph-theoretic
concepts. However, radio transmissions typically decrease
asymptotically with increasing distance. The \emph{physical model}
captures this property much more accurately and is particularly
wide-spread among engineers. Even though the physical model does not
fit in the traditional binary graph-theoretic context, it has similar
properties allowing it to be expressed using edge-weighted conflict graphs.

In this model, network nodes are located in a metric space. The
received signal strength decreases as the distance increases. If a
node transmits at a power level $p$, the signal strength at a distance
of $d$ is $\nicefrac{p}{d^\alpha}$, for a constant $\alpha$. A
transmission is received successfully if ratio of the received signal
strength of the intended transmission and the stengths of concurrent
transmissions plus ambient noise is above some constant
threshold $\beta > 0$. More formally, given pairs of senders $s_i$ and
receivers $r_i$ that transmit at power level $p_i$, receiver $r_i$ can
decode the signal from sender $s_i$ successfully if the SINR
constraint
\[
\frac{p_i}{d(s_i,r_i)^{\alpha}} \ge \beta \left(
\sum_{\substack{j \in M \setminus \{i\}}}
\frac{p_j}{d(s_j,r_i)^{\alpha}} + \nu \right)
\]
is fulfilled. Here $M$ is the set of links transmitting at the same
time on the same channel and $\nu \geq 0$ is a constant expressing
ambient noise. 

Note that we can easily reduce the model to a conflict graph if
transmission powers are fixed. Prominent and simple classes of power
assignments $p\colon V \to \RR_{>0}$ are uniform ($p(v) = 1$) or
linear ($p(v) = d(s_v,r_v)^\alpha$) assignments. More generally, we
can consider assignments satisfying the following monotonicity
constraints. If $d(\ell) \leq d(\ell')$ for two links $\ell, \ell'$
then
\[
p(\ell) \leq p(\ell') \quad \text{ and } \qquad
\frac{p(\ell)}{d(\ell)^\alpha} \geq \frac{p(\ell')}{d(\ell')^\alpha}
\enspace.
\]
We furthermore assume the noise to play a minor role
(cf.~\cite{Kesselheim2010}). 
\begin{proposition}
\label{prop:physical}
The interference constraints in the physical model with fixed
transmission power can be represented by a weighted conflict graph. If
the power assignment satisfies the above constraints, the resulting
\inductiveindependence{} is at most $\O(\log n)$.
\end{proposition}

\begin{proof}
We choose the edges of the conflict graph to have the
following weights. For $\ell = (s, r)$, $\ell' = (s', r')$ we set
\[
w(\ell', \ell) = \min\left\{ 1, \frac{\beta}{1 + \varepsilon} \cdot
  \frac{p(\ell')}{d(s', r)^\alpha} \Bigg/ \left( \frac{p(\ell)}{d(s,
      r)^\alpha} - \frac{\beta}{1 + \varepsilon} \nu \right) \right\}
\enspace,
\]
where
\[
\varepsilon = \frac{\beta}{2} \min_{\ell = (s,r)} \min_{\ell'=(s',r')}
\frac{p(\ell)}{d(s',r)^\alpha} \Bigg/ \frac{p(\ell)}{d(s,r)^\alpha}
\enspace.
\]
By this definition a set $M$ fulfills the SINR constraint iff it
corresponds to an independent set in the edge-weighted graph. The
$\nicefrac{1}{1 + \varepsilon}$ factor is only necessary to get an
exact transformation of the ``$\geq$'' in the SINR condition to the
``$<$'' in the independent set definition. Apart from this factor the
edge weights are equal to the notion of \emph{affectance} $a_p$ in
\cite{Kesselheim2010}, for which we have the following result.
\begin{lemma}[\cite{Kesselheim2010}]
  Let $p$ be a power assignment satisfying Conditions~1 and 2 in
  \cite{Kesselheim2010}

  If $M$ is a set of links that can concurrently transmit and $\ell$
  is link with $d(\ell) \leq d(\ell')$ for all $\ell' \in M$, then
  \[
  \sum_{\ell' \in M} a_p(\ell', \ell) = \O(1) \qquad \text{ and }
  \qquad \sum_{\ell' \in M} a_p(\ell, \ell') = \O(\log n) \enspace.
  \]
\end{lemma}

This immediately yields the edge-weighted graph to have an
\inductiveindependence{} $\rho = \O(\log n)$.
\end{proof}

Interestingly, we can also use our approach if transmission powers are
not given upfront. In this case, our algorithm has to decide about the
assignment of links to channels and which transmission powers to use
for each link. The first part is solved by LP rounding as above. In
the LP we use edge weights ensuring that there is a feasible power
assignment for the computed set of links. The second task of power
assignment can then by done using a power control procedure by
Kesselheim~\cite{Kesselheim2010a}.

Note that, in contrast to the interference models mentioned above, in
this case not all feasible solutions (i.e., feasibly scheduled sets of
links) correspond to independent sets in the weighted graph. However,
for our argument it suffices to observe that each set of feasible
links corresponds to an LP solution for some $\rho$ and that integral
LP solutions with $\rho = 1$ also correspond to feasible sets of
links. Combining these insights with the bounds
in~\cite{Kesselheim2010a} and the ones we proved above, we obtain the
following result.

\begin{theorem}
  \label{theo:physical}
  There is a choice of edge weights such that our algorithm in
  combination with the power control procedure
  in~\cite{Kesselheim2010a} achieves an $\O(\sqrt{k} \log n)$
  approximation in fading metrics and an $\O(\sqrt{k} \log^2 n)$
  approximation in general metrics.
\end{theorem}

\begin{proof}
We define the weighted graph as follows. The set of vertices is again
the set of all links $\mathcal{R}$. The ordering $\pi$ is the ordering
from large to small distances between the sender and its
receiver. Between two links $\ell = (s, r)$ and $\ell'=(s', r')$, we
have the following weight
\[
w(\ell, \ell') = \begin{cases}
  \frac{1}{\tau} \min \left\{ 1, \frac{d(s, r)^\alpha}{d(s, r')^\alpha} \right\} + \frac{1}{\tau} \min \left\{ 1, \frac{d(s, r)^\alpha}{d(s', r)^\alpha} \right\} & \text{if $\pi(\ell) < \pi(\ell')$} \\
  0 & \text{otherwise}
\end{cases} \enspace,
\]
\[
\text{where} \qquad \tau = \frac{1}{2 \cdot 3^\alpha \cdot \left(4
    \beta + 2\right)} \enspace.
\]
Theorem~3 in~\cite{Kesselheim2010a} states that for each independent
set in the weighted graph the power control algorithm calculates a
feasible set of links.

On the other hand Theorem~1 in~\cite{Kesselheim2010a} shows that under
the above edge weights each feasible set of links is also an LP
solution for some $\rho=\O(1)$ in fading metrics. Theorem~7
in~\cite{Kesselheim2010a} shows $\rho=\O(\log n)$ in general metrics.

In conclusion, this implies that by applying our rounding algorithm to
the LP using above defined weights we get a solution, for which we can
apply the power assignment of~\cite{Kesselheim2010a} to obtain a
feasible set of links. The resulting allocation is an $\O(\sqrt{k}
\log n)$ approximation for fading metrics and an $\O(\sqrt{k} \log^2
n)$ approximation in general metrics.
\end{proof}
\section{Mechanism Design}
\label{sec:mechanism}
In this section we show how to apply the framework proposed by Lavi
and Swamy~\cite{Lavi2005} to obtain a truthful mechanism for the
problem, in which the valuations for the allocations are private
information. We only highlight the main ideas of this technique and
the most important observations that allow the use for our problem.

The main idea of the approach is to decompose an optimal LP solution
$x^\ast$ into a set of polynomially many integral solutions with the
following property. For each integral solution we determine a
probability, and the expected social welfare of a randomly chosen solution
according to the probabilities is exactly $b^\ast / \alpha$, where in
our case $\alpha = 8 \cdot \sqrt{k} \cdot \rho$. Given such a
decomposition, we can use scaled VCG payments to implement a
randomized mechanism that is truthful in expectation. For an
accessible presentation of the complete technique,
see~\cite[Chapter 12]{Nisan07} or~\cite{Lavi2005}.

In particular, for simplicity let us first consider only a constant
number of channels; the adjustment to arbitrary many channels is
treated below. We ask the vertices to obtain all valuations for all
channel bundles and solve the corresponding LP (interference
information is assumed to be publicly available). Note that at this
point we are given the optimal solution to an \emph{infeasible} LP. We
set up a decomposition LP with exponentially many variables -- one for
each \emph{feasible} integral solution -- that represent our desired
probabilities. This LP has polynomially many constraints but exponentially
many variables. We can construct the dual with
polynomially many variables and exponentially many constraints. The
variables can be interpreted as valuations in an adjusted
combinatorial auction problem. If this problem has an algorithm that
verifies an integrality gap, we obtain a separation oracle and can
solve the dual decomposition LP in polynomial time. In particular, it
allows us to construct an equivalent LP with a polynomial number of
constraints, i.\,e., the ones corresponding to the solutions obtained
by our algorithm. For this polynomial-sized dual we construct the
primal and determine the polynomially many probabilities of the
solutions found by our algorithm, which completes the decomposition.

It remains to verify that our algorithms provide integral solutions
within the desired integrality gap of $\alpha$ for the adjusted
combinatorial auction problems using dual variables as valuations. We
note here that our algorithms bound the integrality gap of
LP~\eqref{eq:lp} and~\eqref{eq:lp-weighted}, and they can be
derandomized using the technique of pairwise independence. In this
way, given an optimal LP solution $x^\ast$ we can obtain an integral
solution of value at least $b^\ast/\alpha$. Note that our LP describe,
in fact, relaxations of the combinatorial auction problem with
conflict graphs, because Conditions~\eqref{eq:lp:niceproperty}
and~\eqref{eq:lp-weighted:niceproperty} allow each vertex to have
multiple neighbors on the same channel. An arbitrary integral solution
to the LP might thus be infeasible for the original problem. This is
even more severe in the case of the physical model with power control,
where even the interpretation of edge weights is significantly
disconnected from the actual interference that is received. However,
our algorithms produce feasible integral solutions with the desired
gap to the infeasible fractional optimum. Thus, they also prove the
gap for a potential fractional optimum to the LP describing the (more
constrained) exact combinatorial auction problem with conflict graphs
in the respective cases. The remaining arguments can be adapted
from~\cite{Lavi2005} almost without adjustment.

In case of an arbitrary number of channels, we can use demand oracles
to solve the LPs. This results in only a polynomial
number of (non-zero) variables for the LP and for the dual of the
decomposition LP. Note that the procedure to separate the dual of the
decomposition LP does not require demand oracles. In fact, the
complete decomposition procedure can be carried out without accessing
the original bidder valuations.

\section{Asymmetric Channels}
\label{sec:asymmetric}
Up to now, channels were symmetric in terms of interference, which
means the same interference model is applied to each channel. In a
more general setting, for each of the $k$ channels a different edge
set $E_j$ resp. a different edge-weight function $w_j$ for the
interference graph is given.

In this case, we have an edge weight function $\bar{w}_j$ for each
channel $j \in [k]$. The above LP relaxation be easily adapted by
exchanging $\bar{w}$ by $\bar{w}_j$ in the
constraints~\eqref{eq:lp:niceproperty}. In contrast, the analysis of
the rounding algorithms internally depends on the assumption of
symmetric channels. In particular, the proof of
Lemma~\ref{lemma:unweighted-analysis:lostinconflictresolution} uses
the symmetry.

However, when exchanging the probability for a vertex $v$ to choose
set $T$ by $\nicefrac{x_{v, T}^{(l)}}{2 k \rho}$
resp. $\nicefrac{x_{v, T}^{(l)}}{4 k \rho}$, the proof of
Lemma~\ref{lemma:unweighted-analysis:lostinconflictresolution} can be
carried out the same way without using the symmetry.

Hence, for the asymmetric case, we lose a factor of $\O(k \cdot \rho)$
resp. $\O(k \cdot \rho \cdot \log n)$ in the LP rounding step. This
represents our approximation ratio. The result may seem like a trivial
generalization of the $k=1$ case. However, this is not true as
multiple graphs make the problem much harder. We can justify the
approximation factor by a hardness bound.

\begin{theorem}
  For each $\rho$, $k$ with $\rho \cdot k = \O(\log n)$ there is no
  $\nicefrac{\rho \cdot k}{2^{\O(\sqrt{\log(\rho \cdot k)})}}$
  approximation algorithm for asymmetric channels unless \classP =
  \classNP.
\end{theorem}

\begin{proof}
  Again, such an algorithm could be used to approximate the
  independent set problem in bounded-degree graphs. Given a graph
  $G=(V,E)$ with maximum degree $d$, we construct $k$ graphs $G_1 =
  (V, E_1)$, \ldots, $G_k = (V, E_k)$ each having an
  \inductiveindependence{} of $\rho = \nicefrac{d}{k}$. For simplicity
  of notation, we assume this is an integer.

  Let $\{v_1, \ldots, v_n\}$ be an arbitrary ordering of the
  vertices. We now distribute the edges from $E$ to the edge sets
  $E_1$, \ldots, $E_k$. For a vertex $v_i$ the incident edges to
  vertices $v_j$ of lower index are distributed such that each edge
  set gets at most $\rho$ such edges. Since the maximum vertex degree
  is $d$ this is always possible. The valuations for the vertices are
  chosen such that for all vertices $b_{v, T}$ is $1$ only for $T=[k]$
  and $0$ otherwise.

  By this construction allocations of valuation $b$ exactly correspond
  to independent sets of size $b$. Thus, such an approximation
  algorithm cannot exist unless \classP = \classNP.
\end{proof}

As we see, for asymmetric channels our algorithms are close to optimal
without making further assumptions about the interference model.

\section{Open Problems}
\label{sec:open}
In this paper we present a general framework for secondary spectrum
auctions that works with a large number of interference models. Our
approach can easily be extended to even more models by proving bounds
on the \inductiveindependence{} in the associated graphs. To improve
the results in this paper, it would, e.g., be interesting to know if
for the physical model it also holds that $\rho = \O(1)$ in general
metrics or for distance-based power assignments.

For obtaining a truthful mechanism we use decomposition and rounding
of LP solutions, and we heavily rely on the ellipsoid method. It is an
interesting question if this could be avoided to make the algorithm
more applicable in practice.

\thispagestyle{empty}
\bibliographystyle{plain}
\bibliography{bibliography}

\end{document}